\documentclass[letterpaper, 11pt]{article}
\usepackage{amsmath,amsfonts,amssymb,mathrsfs}
\usepackage{latexsym}
\usepackage{amsthm} 
\usepackage{enumerate}
\usepackage{algorithm, algorithmicx, algpseudocode}
\usepackage{relsize}
\usepackage{lineno}
\usepackage[margin=1.6in]{geometry}

\newtheorem{theorem}{Theorem}
\newtheorem{lemma}{Lemma}
\newtheorem{proposition}{Proposition}
\newtheorem{corollary}{Corollary}

\newcommand{\calA}{\mathcal{A}}

\newcommand{\calF}{\mathcal{F}}

\newcommand{\calG}{\mathcal{G}}

\newcommand{\calI}{\mathcal{I}}

\newcommand{\N}{\mathbb{N}}

\newcommand{\R}{\mathbb{R}}

\newcommand{\calT}{\mathcal{T}}

\newcommand{\Z}{\mathbb{Z}}

\DeclareMathOperator*{\supp}{supp}

\newcommand{\Bernoulli}{\text{Bernoulli}}
\newcommand{\Geometric}{\text{Geometric}}
\newcommand{\ps}{\sigma}
\newcommand{\disj}{\textsc{Disj}}
\newcommand{\indx}{\textsc{Index}}

\begin{document}

\title{Universal sketches for the frequency negative moments and other decreasing streaming sums}
\author{Vladimir Braverman\thanks{Computer Science, Johns Hopkins University. E-mail:\texttt{vova@cs.jhu.edu}}~~and Stephen R.\ Chestnut\thanks{Applied Mathematics and Statistics, Johns Hopkins University.  E-mail:\texttt{schestn2@jhu.edu}}}

\maketitle
\begin{abstract}
  Given a stream with frequencies $f_d$, for $d\in[n]$, we characterize the space necessary for approximating the frequency negative moments~$F_p=\sum |f_d|^p$, where $p<0$ and the sum is taken over all items $d\in[n]$ with nonzero frequency,  in terms of $n$, $\epsilon$, and $m=\sum |f_d|$.
  To accomplish this, we actually prove a much more general result.
  Given any nonnegative and nonincreasing function $g$, we characterize the space necessary for any streaming algorithm that outputs a $(1\pm\epsilon)$-approximation to $\sum g(|f_d|)$,  where again the sum is over items with nonzero frequency.
  The storage required is expressed in the form of the solution to a relatively simple nonlinear optimization problem, and the algorithm is universal for $(1\pm\epsilon)$-approximations to any such sum where the applied function is nonnegative, nonincreasing, and has the same or smaller space complexity as $g$.
  This partially answers an open question of Nelson~(IITK Workshop Kanpur, 2009).
\end{abstract}



\section{Introduction}


A \emph{stream} is a sequence $S = ((d_1,\delta_1),(d_2,\delta_2),\ldots,(d_{N},\delta_{N}))$, where $d_i\in[n]$ are called the items or elements in the stream and $\delta_i\in\Z$ is an update to the $d_i$th coordinate of an implicitly defined $n$-dimensional vector.  
Specifically, the \emph{frequency} of $d\in[n]$ after $k\leq N$ updates is 
\[f^{(k)}_d = \sum \{\delta_j| j\leq k, d_j=d\},\] and the implicitly defined vector $f:= f^{(N)}$ is commonly referred to as the \emph{frequency vector} of the stream $S$.  
Let $M:=\max\{n,|f^{(k)}_d|:d\in[n],0\leq k\leq N\}$ and $m=\sum_d |f_d|$; thus, it requires $O(\log{M})$ bits to exactly determine the frequency of a single item.
This model is commonly known as the \emph{turnstile} streaming model, as opposed to the \emph{insertion-only} model which has $\delta_i=1$, for all $i$, but is the same otherwise.
In an insertion-only stream $N=m\geq M$.

Streams model computing scenarios where the processor has very limited access to the input.
The processor reads the updates one-at-a-time, without control of their order, and is tasked to compute a function on the frequency vector.
The processor can perform its computation exactly if it stores the entire vector $f$, but this may be undesirable or even impossible when the dimension of $f$ is large. 
Thus, the goal is to complete the computation using as little storage as possible.
Typically, exact computation requires storage linear in $n$, so we seek approximations.

Given a stream with frequencies $f_d$, for $d\in[n]$, we consider the problem of approximating the frequency negative moments, specifically $F_p=\sum |f_d|^p$ where $p<0$ and the sum is taken over all items $d\in[n]$ with nonzero frequency.
We characterize, up to factors of $O(\epsilon^{-1}\log^2{n}\log{M})$ in the turnstile model and $O(\epsilon^{-1}\log{M})$ in the insertion-only model, the space necessary to produce a $(1\pm\epsilon)$-approximation to $F_p$, for $p<0$, in terms of the accuracy $\epsilon$, the dimension $n$, and the $L^1$ length $m$ of $f$.  

Negative moments, also known as ``inverse moments'', of a probability distribution have found several applications in statistics.
Early on, they were studied in application to sampling and estimation problems where the sample size is random~\cite{stephan1945expected,grab1954tables} as well as in life-testing problems~\cite{mendenhall1960approximation}. 
More recently, they appear in the design of multi-center clinical trials \cite{jones2004approximating} and in the running time analysis of a quantum adiabatic algorithm for $3$-\textsc{SAT}~\cite{znidaric2005asymptotic,znidaric2006exponential}.
$F_0/F_{-1}$ is the harmonic mean of the (nonzero) frequencies in the insertion-only model, and more generally, the value $(F_p/F_0)^{1/p}$ is known as the $p$th power mean~\cite{bullen2003handbook}.  
The harmonic mean is the truest average for some types of data, for example speeds, parallel resistances, and P/E~ratios~\cite{reilly2004handbook}.

To our knowledge this is the first paper to consider streaming computation of the frequency negative moments and the first to determine the precise dependence of the space complexity of streaming computations on $m$.
In fact, in the process of characterizing the storage necessary to approximate the frequency negative moments, we actually characterize the space complexity of a much larger class of streaming sum problems.
Specifically, given any nonnegative, nonincreasing function $g:\N\to\R$ we determine to within a factor of $O(\epsilon^{-1}\log^2{n}\log {M})$ the space necessary to approximate
\[g(f):=\sum_{d\in\supp(f)}g(|f_d|),\]
where $\supp(f):=\{d\in[n]:f_d\neq0\}$ is the support of $f$.
Furthermore, the sketch providing a $(1\pm\epsilon)$-approximation for $g(f)$ is universal for a $(1\pm\epsilon)$-approximation for any nonnegative nonincreasing function with the same or smaller space complexity as $g$.
This partially answers a question of Nelson~\cite{sublinear_open_30} -- which families of functions admit universal sketches?

The attention on $m$ is warranted; in fact, the complexity in question depends delicately on this parameter.
If we forget about $m$ for a moment, then a standard reduction from the communication problem $\indx$ implies that computing a $(1\pm\frac{1}{2})$-approximation to $F_p$, for $p<0$, requires $\Omega(n)$ bits of storage -- nearly enough to store the entire vector $f$.
However, the reduction requires $m=\Omega(n^{1-1/p})$, recall that $p<0$.
If $m=o(n^{1-1/p})$ then, as we show, one can often get away with $o(n)$ bits of memory.

The next two sections outline our approach to the decreasing streaming sum problem and state our main results.
Section~\ref{sec: background} reviews previous work on streaming sum problems.
In Section~\ref{sec: frequency negative moments} we show how our results solve the frequency negative moments problem.
Sections~\ref{sec: lower bounds} and~\ref{sec: upper bounds} prove the main results.
Finally, Section~\ref{sec: computing ps} and Appendix~\ref{app: details} describe the implementation details for the streaming setting.

\subsection{Preliminaries}

Let $\calF = \{f\in\N^n:\sum f_d\leq m\}$ and let $\calT$ and $\calI$ denote the sets of turnstile streams and insertion-only streams, respectively, that have their frequency vector $f$ satisfying $|f|\in\calF$.
The set $\calF$ is the set of all nonnegative frequency vectors with $L^1$~norm at most~$m$.
Clearly, $\calF$ is the image under coordinate-wise absolute value of the set of all frequency vectors with $L^1$ norm at most $m$.
We assume $n\leq m$. 

In order to address the frequency negative moments problem we will address the following more general problem.
Given a nonnegative, nonincreasing function $g:\N\to\R$, how much storage is needed by a streaming algorithm that $(1\pm\epsilon)$-approximates $g(f)$, for the frequency vector $f$ of any stream $S\in\calT$ or $S\in\calI$?
Equivalently, we can assume that $g(0)=0$, $g$ is nonnegative and nonincreasing on the interval $[1,\infty)$, and extend the domain of $g$ to $\Z$ by requiring it to be symmetric, i.e., $g(-x)=g(x)$.
  Therefore, $g(f) = \sum_{i=1}^ng(f_d)$.
For simplicity, we call such functions ``decreasing functions''.

A randomized algorithm $\calA$ is a \emph{turnstile streaming $(1\pm\epsilon)$-approximation algorithm} for $g(f)$ if
\[P\left\{(1-\epsilon)g(f)\leq \calA(S)\leq (1+\epsilon)g(f)\right\}\geq \frac{2}{3}\] 
holds for every stream $S\in\calT$, and insertion only algorithms are defined analogously.
For brevity, we just call such algorithms ``approximation algorithms'' when $g$, $\epsilon$, and the streaming model are clear from the context.
We consider the maximum number of bits of storage used by the algorithm $\calA$ with worst case randomness on any valid stream.

A sketch is a, typically randomized, data structure that functions as a compressed version of the stream.
Let $\calG\subseteq\R^\N\times(0,1/2]$.
We say that a sketch is \emph{universal} for a class $\calG$ if for every $(g,\epsilon)\in\calG$ there is an algorithm that, with probability at least $2/3$, extracts from the sketch a $(1\pm\epsilon)$-approximation to $g(f)$.
The probability here is taken over the sketch as well as the extraction algorithm.


Our algorithms assume a priori knowledge of the parameters $m$ and $n$, where $m=\|f\|_1$ and $n$ is the dimension of $f$.
In practice, one chooses $n$ to be an upper bound on the number of distinct items in the stream.
Our algorithm remains correct if one instead only knows $m\geq \|f\|_1$, however if $m\gg\|f\|_1$ the storage used by the algorithm may not be optimal.
We assume that our algorithm has access to an oracle that computes $g$ on any valid input.
In particular, the final step of our algorithms is to submit a list of frequencies, i.e., a sketch, as inputs for $g$.
We do not count the storage required to evaluate $g$ or to store its value.

\subsection{Our results}\label{sec: our results}



Our lower bound is proved by a reduction from the communication complexity of disjointness wherein we parameterize the reduction with the coordinates of $|f|$, the absolute value of a frequency vector.
The parameterization has the effect of giving a whole collection of lower bounds, one for each frequency vector among a set of many.
Specifically, if $f\in\calF$ and $g(f)\leq \epsilon^{-1}g(1)$ then we find an $\Omega(|\supp(f)|)$ lower bound on on the number of bits used by any approximation algorithm.
This naturally leads us to the following nonlinear optimization problem
\begin{equation}\label{eq: ps definition}
\ps(\epsilon,g,m,n):=\max\left\{|\supp(f)|: f\in\calF, g(f)\leq\epsilon^{-1}g(1)\right\},
\end{equation}
which gives us the ``best'' lower bound.
We will use $\ps=\ps(\epsilon,g,m,n)$ when $\epsilon$, $g$, $m$, and $n$ are clear from the context.
Our main lower bound result is the following.

\begin{theorem}\label{thm: decreasing lower bound}
Let $g$ be a decreasing function, then any $k$-pass insertion-only streaming $(1\pm\epsilon)$-approximation algorithm requires $\Omega(\ps/k)$ bits of space.
\end{theorem}

Before we consider approximation algorithms, let us consider a special case. Suppose there is an item~$d^*$ in the stream that satisfies $g(f_{d^*})\geq\epsilon g(f)$.
An item such as $d^*$ is called an $\epsilon$-heavy hitter.
If there is an $\epsilon$-heavy hitter in the stream, then $g(1)\geq g(f_{d^*})\geq\epsilon g(f)$ which implies $|\supp(f)|\leq\ps$, by the definition of $\ps$.
Of course, in this case it is possible compute $g(f)$ with $O(\ps\log{M})$ bits in one pass in the insertion-only model and with not much additional space in the turnstile model simply by storing a counter for each element of $\supp(f)$.
Considering the $\Omega(\ps)$ lower bound, this is nearly optimal.
However, it only works when $f$ contains an $\epsilon$-heavy hitter.

Our approximation algorithm is presented next.
It gives a  uniform approach for handling all frequency vectors, not just those with $\epsilon$-heavy elements.
\begin{algorithm}
  \begin{algorithmic}[1]
    \State Compute $\ps=\ps(\epsilon,g,m,n)$ and let
    \begin{equation}\label{eq: sampling probability}
      q\geq\min\left\{1,\frac{9\ps}{\epsilon|\supp(f)|}\right\}.
    \end{equation}
    \State Sample pairwise independent random variables $X_d\sim\Bernoulli(q)$, for $d\in[n]$, and let $W=\{d\in\supp(f):X_d=1\}$.
    \State Compute $f_d$, for each $d\in W$.
    \State Output $q^{-1}\sum_{d\in W}g(f_d)$.
  \end{algorithmic}
  \caption{$(1\pm\epsilon)$-approximation algorithm for $g(f)$.}
\label{algo: approximate}
\end{algorithm}

Algorithm~\ref{algo: approximate} simply samples each element of $\supp(f)$ pairwise independently with probability $q$.
The expected sample size is $q|\supp(f)|$, so in order to achieve optimal space we take equality in Equation~\ref{eq: sampling probability}.
The choice yields, in expectation, $q|\supp(f)|=O(\ps/\epsilon)$ samples.
Section~\ref{sec: computing ps} and Appendix~\ref{app: details} explain how to implement the algorithm for the streaming setting and the correctness is established by the following theorem.
It is proved in Section~\ref{sec: upper bounds}. 
\begin{theorem}\label{thm: decreasing sketches}
There is a turnstile streaming algorithm that, with probability at least $2/3$, outputs a $(1\pm\epsilon)$-approximation to $g(f)$ and uses $O(\epsilon^{-1}\ps\log^2(n)\log(M))$ bits of space.
The algorithm can be implemented in the insertion-only model with $O(\epsilon^{-1}\ps\log(M)+\log^2{n})$ bits of space.
\end{theorem}

It is worthwhile to remark that the suppressed constants in the asymptotic bounds of Theorems~\ref{thm: decreasing lower bound} and~\ref{thm: decreasing sketches} are independent of $g$, $\epsilon$, $m$, and $n$.

The optimization problem \eqref{eq: ps definition} reappears in the proof of Theorem~\ref{thm: decreasing sketches}.
The key step is the observation mentioned above.
Namely, for the particular frequency vector~$f$ that is our input, if there is an item $d$ satisfying $g(|f_d|)\geq \epsilon g(f)$ then $|\supp(f)|\leq\ps$.

Let us now emphasize a particular feature of this algorithm.
Previously, we commented that choosing equality in \eqref{eq: sampling probability} is optimal in terms of the space required.
However, Algorithm~\ref{algo: approximate} is still correct when the inequality is strict.
Notice that the sketch is just a (pairwise independent) random sample of $\supp(f)$ and its only dependence on $g$ and $\epsilon$ is through the parameter $\ps/\epsilon$.
Let $g'$ and $\epsilon'$ be another decreasing function and error parameter satisfying $\frac{\sigma(\epsilon',g',m,n)}{\epsilon'} \leq \frac{\sigma(\epsilon,g,m,n)}{\epsilon}$, then
\[q' = \min\left\{1,\frac{9\ps'}{\epsilon'|\supp(f)|}\right\}\leq q = \min\left\{1,\frac{9\ps}{\epsilon|\supp(f)|}\right\}.\]
In particular, this means that the sketch that produces an $(1\pm\epsilon)$-approximation to $g(f)$ also suffices for an $(1\pm\epsilon')$-approximation to $g'$.
For example, if one takes $g'\geq g$, pointwise with $g'(1)=g(1)$, then $\sigma(\epsilon,g',m,n)\leq\sigma(\epsilon,g,m,n)$ so one can extract from the sketch $(1\pm\epsilon)$-approximations to $g(f)$ and $g'(f)$, each being separately correct with probability $2/3$. 
Thus, the sketch is universal for any decreasing function $g'$ and accuracy $\epsilon'$ where $\ps(\epsilon',g',m,n)\leq\ps(\epsilon,g,m,n)$.
In the context of the frequency negative moments, this implies that the sketch yielding a $(1\pm\epsilon)$-approximation to $F_p$, for $p<0$, is universal for $(1\pm\epsilon)$-approximations of $F_{p'}$, for all $p<p'<0$.

Computing the sketch requires a priori knowledge of $\ps$.
If one over-estimates $\ps$ the algorithm remains correct, but the storage used increases. 
To know $\ps$ requires knowledge of $m$, or at least an good upper bound on $m$.
This is a limitation, but there are several ways to mitigate it.
If one does not know $m$ but is willing to accept a second pass through the stream, then using the algorithm of \cite{kane2010exact} one can find a $(1\pm\frac{1}{2})$-approximation to $m$ with $O(\log{M})$ bits of storage in the first pass and approximate $g(f)$ on the second pass.
A $(1\pm\frac{1}{2})$-approximation to $m$ is good enough to determine $\ps$ to within a constant, which is sufficient for the sketch.
Alternatively, one can decide first on the space used by the algorithm and, in parallel within one pass, run the algorithm and approximate $m$.
After reading the stream one can determine for which decreasing functions $g$ and with what accuracy~$\epsilon$ does the approximation guarantee hold.

\subsection{Background}\label{sec: background}

Much of the effort dedicated to understanding streaming computation, so far, has been directed at the frequency moments $F_p = \sum |f_i|^p$, for $0<p<\infty$, as well as $F_0$ and $F_\infty$, the number of distinct elements and the maximum frequency respectively.  
In the turnstile model, $F_0$ is distinguished from $L^0=|\supp(f)|$, the number of elements with a nonzero frequency.

The interest in the frequency moments began with the seminal paper of Alon, Matias, and Szegedy~\cite{alon1996space}, who present upper and lower bounds of $O(\epsilon^{-2}n^{1-1/p})$ and $\Omega(n^{1-5/p})$, respectively, on the space needed to find a $(1\pm\epsilon)$-approximation to $F_p$, and a separate $O(\epsilon^{-2}\log{m})$ space algorithm for $F_2$.
Since then, many researchers have worked to push the upper and lower bounds closer together. 
We discuss only a few of the papers in this line of research, see \cite{woodruff2014data} an the references therein for a more extensive history of the frequency moments problem.

To approximate $F_p$, Alon, Matias, and Szegedy inject randomness into the stream and then craft an estimator for $F_p$ on the randomized stream.
A similar approach, known as stable random projections, is described by Indyk~\cite{indyk2006stable} for $F_p$, when $0<p\leq 2$ (also referred to as $\ell_p$ approximation).
Kane, Nelson, and Woodruff~\cite{kane2010exact} show that Indyk's approach, with a more careful derandomization, is optimal.
Using the method of stable random projections, Li~\cite{li2008estimators} defined the so-called \emph{harmonic mean estimator} for $F_p$, when $0<p<2$, which improves upon the sample complexity of previous methods.
We stress that this is not an estimator for the harmonic mean of the frequencies in a data stream, rather it is an estimator for $F_p$ that takes the form of the harmonic mean of a collection of values.

For $p>2$, the AMS approach was improved upon~\cite{coppersmith2004improved,ganguly2004estimating} until a major shift in the design of streaming algorithms began with the algorithm of Indyk and Woodruff~\cite{indyk2005optimal} that solves the frequency moments problem with, nearly optimal,  $n^{1-2/p}(\frac{1}{\epsilon}\log{n})^{O(1)}$ bits.
Their algorithm introduced a recursive subsampling technique that was subsequently used to further reduce space complexity \cite{bhuvanagiri2006simpler,braverman2010recursive}, which now stands at $O(\epsilon^{-2}n^{1-2/p}\log{n})$ in the turnstile model~\cite{ganguly2011polynomial} with small $\epsilon$ and $O(n^{1-2/p})$ in the insertion-only model with $\epsilon=\Omega(1)$~\cite{braverman2014approximating}.

Recently, there has been a return to interest in AMS-type algorithms motivated by the difficulty of analyzing algorithms that use recursive subsampling.
``Precision Sampling'' of Andoni, Krauthgamer, and Onak~\cite{andoni2011streaming} is one such algorithm that accomplishes nearly optimal space complexity without recursive subsampling.  
Along these lines, it turns out that one can approximate $g(f)$ by sampling elements~$d\in[n]$ with probability roughly $q_d \approx g(f_d)/\epsilon^2 g(f)$, or larger, and then averaging and scaling appropriately, see Proposition~\ref{prop: sampling archetype}.
Algorithm~\ref{algo: approximate} takes this approach, and also fits in the category of AMS-type algorithms.
However, it is far from clear how to accomplish this sampling optimally in the streaming model for a completely arbitrary function $g$.

A similar sampling problem has been considered before.
Monemizadeh and Woodruff~\cite{monemizadeh20101} formalized the problem of sampling with probability $q_d = g(f_d)/ g(f)$ and then go on to focus on $L_p$ sampling, specifically $g(x)=|x|^p$, for $0\leq p\leq 2$.
In follow-up work, Jowhari, S\u aglam, and Tardos offer $L_p$ sampling algorithms with better space complexity~\cite{jowhari2011tight}.


As far as the frequency moments lower bounds go, there is a long line of research following AMS~\cite{baryossef2002information,chakrabarti2003near,gronemeier2009asymptotically,andoni2013tight} that has led to a lower bound matching the best known turnstile algorithm of Ganguly~\cite{ganguly2011polynomial} to within a constant~\cite{li2013tight}, at least for some settings of $m$ and $\epsilon$.
The insertion-only algorithm of Braverman et al.~\cite{braverman2014approximating} matches the earlier lower bound of Chakrabarti, Khot, and Sun~\cite{chakrabarti2003near}.

For a general function~$g$ not much is known about the space-complexity of approximating $g(f)$.  
Most research has focused on specific functions.
Chakrabarti, Do Ba, and Muthukrishnan~\cite{chakrabarti2006estimating} and Chakrabarti, Cormode, and Muthukrishnan~\cite{chakrabarti2007near} sketch the Shannon Entropy. Harvey, Nelson, and Onak~\cite{harvey2008sketching} approximate Renyi~$\log(\|f\|_{\alpha}^\alpha)/(1-\alpha)$, Tsallis~$(1-\|x\|_{\alpha}^\alpha)/(\alpha-1)$, and Shannon entropies.
Braverman, Ostrovsky, and Roytman~\cite{braverman2010zero,braverman2014universal} characterized nonnegative, nondecreasing functions that have polylogarithmic-space approximation algorithms present a universal algorithm, based on the subsampling technique, for the same.
Guha, Indyk, McGregor~\cite{guha2007sketching} study the problem of sketching common information divergences between the streams, i.e., statistical distances between the probability distributions with p.m.f.s $e/\|e\|_1$ and $f/\|f\|_1$.


\section{The frequency negative moments}\label{sec: frequency negative moments}

Before proving Theorems~\ref{thm: decreasing lower bound} and~\ref{thm: decreasing sketches}, let us deploy them to determine the streaming space complexity of the frequency negative moments.
It will nicely illustrate the trade-off between the length of the stream and the space complexity of the approximation.

The first step is to calculate $\sigma(\epsilon,g,m,n)$, where $g(x)=|x|^p$, for $x\neq0$ and $p<0$, and $g(0)=0$.
There is a maximizer of \eqref{eq: ps definition} with $L^1$ length $m$ because $g$ is decreasing.
The convexity of $g$ implies that $\sigma \leq \max\{s\in\R: s(m/s)^p\leq \epsilon^{-1}\}$, and $\ps$ is at least the minimum of $n$ and $\max\{s\in\N: s(m/s)^p\leq \epsilon^{-1}\}$ by definition. 
Thus, we can take $\sigma =\min\left\{n,\theta\left( \epsilon^{\frac{-1}{1-p}}m^{\frac{-p}{1-p}} \right)\right\}$.
This gives us the following corollary to Theorems~\ref{thm: decreasing lower bound} and~\ref{thm: decreasing sketches}.
\begin{corollary}
  Let $p<0$.
  Any $(1\pm\epsilon)$-approximation algorithm for $F_p$ requires $\Omega(\min\{n,\epsilon^{\frac{-1}{1-p}}m^{\frac{-p}{1-p}}\})$ bits of space.
  Such an approximation can be found with $O(\epsilon^{-\frac{2-p}{1-p}}m^{\frac{-p}{1-p}}\log^2{n}\log{M})$ bits in a turnstile stream and $O(\epsilon^{-\frac{2-p}{1-p}}m^{\frac{-p}{1-p}}\log{M})$ bits in an insertion-only stream.
\end{corollary}

For example, taking $p=-1$ we find that the complexity is approximately $\frac{\ps}{\epsilon} = \min\{n,\theta(\epsilon^{-3/2}m^{1/2})\}$.
This is also the space complexity of approximating the harmonic mean of the nonzero frequencies.
It is apparent from the formula that the relationship between $m$ and $n$ is important for the complexity.

\section{Lower bounds for decreasing streaming sums}\label{sec: lower bounds}

It bears repeating that if $g(x)$ decreases to $0$ as $x\to\infty$ then one can always prove an $\Omega(n)$ lower bound on the space complexity of approximating $g(f)$.
However, the stream needed for the reduction may be very long (as a function of $n$).
Given only the streams in $\calT$ or $\calI$, those with $L^1$-length $m$ or less, a weaker lower bound may be the best available.
The present section proves this ``best'' lower bound, establishing Theorem~\ref{thm: decreasing lower bound}.
 
The proof uses a reduction from the communication complexity of disjointness, see the book of Kushilevitz and Nisan~\cite{kushilevitz1996communication} for background on communication complexity.
The proof strategy is to parameterize the lower bound reduction in terms of the frequencies $f$.
Optimizing the parameterized bound over $f\in\calF$ gives the best possible bound from this reduction.

The proof of Theorem~\ref{thm: decreasing lower bound} is broken up with a two lemmas.
The first lemma is used in the reduction from $\disj(s)$, the $s$-element disjointness communication problem.
It will show up again later when we discuss a fast scheme for computing $\ps$ for general functions.

\begin{lemma}\label{lem: decreasing lower bound booster}
Let $y_i\in\R_{\geq 0}$, for $i\in[s]$, and let $v:\R\to\R_{\geq 0}$. 
If $\sum y_i\leq Y$ and $\sum v(y_i)\leq V$, then there exists $i$ such that $\frac{s}{2} y_i\leq Y$ and $\frac{s}{2} v(y_i)\leq V$.
\end{lemma}
\begin{proof}
Without loss of generality $y_1\leq y_2\leq\cdots \leq y_s$.  
Let $i_j$, $j\in [\ps]$, order the sequence such that $v(y_{i_1})\leq v(y_{i_2})\leq \cdots\leq v(y_{i_s})$ and let $I = \{i_j|j\leq \lfloor s/2\rfloor +1\}$.  
By the Pigeon Hole Principle, there exists $i\in I$ such that $i\leq \lfloor s/2\rfloor +1$.  
Thus $\frac{s}{2}y_i\leq \sum_{j=\lfloor s/2\rfloor+1}^s y_{i_j} \leq Y$ and $\frac{s}{2} v(y_i)\leq \sum_{j=\lfloor s/2\rfloor+1}^s v(y_j) \leq V$.
\end{proof}

\begin{lemma}\label{lem: decreasing lower bound}
Let $g$ be decreasing and $\epsilon>0$.
If $f = (y,y,\ldots,y,0,\ldots,0)\in\calF$ and $g(f)\leq\epsilon^{-1}g(1)$,
then any $k$-pass $(1\pm\epsilon)$-approximation algorithm requires $\Omega(|\supp(f)|/k)$ bits of storage.
\end{lemma}
\newenvironment{proofsketch}{\paragraph{Sketch of proof.}}{}
\begin{proofsketch}
Let $\calA$ be an $(1\pm\epsilon)$-approximation algorithm.
We use a reduction from the communication complexity of $\disj(s)$, where $s=\lfloor |\supp(f)|/2\rfloor$.
Alice is given $A\subseteq [s]$ and Bob is given $B\subseteq [s]$.
They jointly create a stream~$S$ with $s$~or~fewer distinct elements such that all of the frequencies are $1$, $y$, or~$y+1$, then they compute the approximation $\calA(S)$ and compare the outcome to a threshold.
Computing $\calA(s)$ requires them to transmit the memory $O(k)$ times.

The number of items in $S$ with frequency $1$ is  $|A\cap B|$, so it can be arranged that when the intersection is empty $\calA(S)$ is smaller than the threshold and otherwise it is larger.
The condition $g(f)\leq\epsilon^{-1}g(1)$ guarantees sufficient separation between the two cases.
We defer the complete proof to Appendix~\ref{app: lower bound proof}.
\end{proofsketch}




\newenvironment{lowerboundproof}{\paragraph{Proof of Theorem~\ref{thm: decreasing lower bound}.}}{}
\begin{lowerboundproof}
Let $f\in\calF$ be a maximizer of \eqref{eq: ps definition} and apply Lemma~\ref{lem: decreasing lower bound booster} to the positive elements of $f$.  
From this we find that there exists $y$ such that $ys'\leq \|f\|_1$ and $g(1)\geq\epsilon s'g(y)$, for $s' = \ps/2$.
Therefore, $f'=(y,y,\ldots,y,0,\ldots,0)\in\calF$ with $\lfloor s'\rfloor$ coordinates equal to $y$.
Applying Lemma~\ref{lem: decreasing lower bound} to $f'$ implies the desired bound.
\end{lowerboundproof}

With Lemma~\ref{lem: decreasing lower bound booster} in mind, one may ask: why not restrict the maximization problem in \eqref{eq: ps definition}, the definition of $\ps$, to  streams that have all frequencies equal and still get the same order lower bound?
This is valid alternative definition. 
In fact, doing so does appreciably affect the effort needed to compute $\ps$, it is one of the main steps used by our algorithm to approximate $\ps$ in Section~\ref{sec: computing ps}.
However, it makes reasoning about $\ps$ a bit messier.
For example, in Section~\ref{sec: our results} we comment that if the frequency vector $f$ contains an $\epsilon$-heavy element then $|\supp(f)|\leq\ps$.
This comes directly from the fact that $\{f'\in\calF : g(f')\leq\epsilon^{-1}g(1)\}$ is the feasible set for \eqref{eq: ps definition}. 
If we restrict the feasible set, then we cannot so directly draw the conclusion. 
Rather, we must compare $g(f)$ to points in the restricted feasible set by again invoking Lemma~\ref{lem: decreasing lower bound booster}.

\section{Correctness of the algorithm}\label{sec: upper bounds}

This section presents the proof that our approximation algorithm is correct.
Algorithm~\ref{algo: approximate} describes the basic procedure, and Appendix~\ref{app: details} describes how it can be implemented in the streaming setting.
The correctness relies on our ability to perform the sampling and the following simple proposition.

\begin{proposition}\label{prop: sampling archetype}
Let $g$ be a nonnegative function and let $X_d\sim\Bernoulli(p_d)$ be pairwise independent random variables with $p_d\geq \min\left\{1,\frac{9g(f_d)}{\epsilon^2 g(f)}\right\}$, for all $d\in[n]$.
Let $\hat{G} = \sum_{d=1}^np_d^{-1}X_dg(f_d)$, then $P(|\hat{G}-g(f)|\leq \epsilon g(f))\geq \frac{8}{9}$.
\end{proposition}
\begin{proof}
We have $E\hat{G}=g(f)$ and $Var(\hat{G})\leq\sum_d p_{d}^{-1}g(f_d)^2 = \frac{1}{9}(\epsilon g(f))^2$, by pairwise independence.
The proposition now follows from Chebyshev's inequality.
\end{proof}

The algorithm samples each element of $\supp(f)$ with probability approximately $\ps/\epsilon\supp(f)$.
In order to show that this sampling probability is large enough for Proposition~\ref{prop: sampling archetype} we will need one lemma; its proof is given in Appendix~\ref{app: details}.
It gives us some control on $\ps(\epsilon,g,m,n)$ as $\epsilon$ varies.

\begin{lemma}\label{lem: ps relations}
If $\alpha<\epsilon$, then $\epsilon (1+\ps(\epsilon,g,m,n))\geq\alpha\ps(\alpha,g,m,n) $.
\end{lemma}

For brevity, we only state here the correctness of the streaming model sampling algorithm, which uses standard techniques.
The details of the algorithm are given in the Appendix~\ref{app: details}.
\begin{lemma}\label{lem: sampling algorithm}
Given $s\leq n$, there is an algorithm using $O(s\log^2{n}\log{M})$ bits of space in the turnstile model and $O(s\log{M}+\log^2{n})$ bits in the insertion-only model that samples each item of $\supp(f)$ pairwise-independently with probability at least $\min\{1,s/|\supp(f)|\}$ and, with probability at least $7/9$, correctly reports the frequency of every sampled item and the sampling probability.
\end{lemma}

Finally, we prove the correctness of our approximation algorithm.
Here is where we will again use the optimality of $\ps$ in its definition \eqref{eq: ps definition}.
In regards to the lower bound of Theorem~\ref{thm: decreasing lower bound}, this upper bound leaves gaps of $O(\epsilon^{-1}\log^2{n}\log{M})$ and $O(\epsilon^{-1}\log{M})$ in the turnstile and insertion-only models, respectively.

\newenvironment{upperboundproof}{\paragraph{Proof of Theorem~\ref{thm: decreasing sketches}.}}{}
\begin{upperboundproof}
  We use the algorithm of Lemma~\ref{lem: sampling algorithm} to sample with probability at least $q= \min\{1,9(\ps+1)/\epsilon|\supp(f)|\}$.
  Let us first assume that $q\geq \min\{1, 9g(f_d)/\epsilon^2 g(f)\}$, for all $d$, so that the hypothesis for Proposition~\ref{prop: sampling archetype} is satisfied. 
  The algorithm creates samples $W_i$, for $i=0,1,\ldots,O(\log{n})$, where each item is sampled in $W_i$ with probability $q_i=2^{-i}$.
  For each $i$ such that $q_i\geq q$, Proposition~\ref{prop: sampling archetype} guarantees that $\hat{G}_i=q_i^{-1}\sum_{d\in W_i}g(f_d)$ is a $(1\pm\epsilon)$-approximation with probability at least $8/9$.
  With probability at least $7/9$, the algorithm returns one of these samples correctly, and then the approximation guarantee holds.
  Thus, the approximation guarantee holds with probability at least $2/3$.

  It remains to show that $q\geq \min\{1,9g(f_d)/\epsilon g(f)\}$, for all $d\in[n]$.
  Let $\alpha = g(1)/g(f)$ then define $\sigma_\epsilon = \sigma(\epsilon,g,m,n)$ and $\sigma_\alpha=\sigma(\alpha,g,m,n)$.
  By definition $|\supp(f)|\leq\sigma_\alpha$, thus if $\alpha \geq\epsilon$ then $|\supp(f)|\leq \sigma_\alpha\leq \sigma_\epsilon$, so the sampling probability is 1 and the claim holds.
 
Suppose that $\alpha<\epsilon$.  
For all $d\in[n]$, we have
\begin{align*}
\frac{g(f_d)}{g(f)}\leq\frac{g(1)}{g(f)} = \alpha \leq \frac{\epsilon(1+\ps_\epsilon)}{\ps_\alpha} \leq \frac{\epsilon(1+\ps_\epsilon)}{|\supp(f)|},
\end{align*}
where the second inequality comes from Lemma~\ref{lem: ps relations} and the third from the definition of $\ps_\alpha$ as a maximum.
In particular, this implies that
\[\frac{9\epsilon^{-1}(\ps+1)}{|\supp(f)|}\geq \frac{9g(f_d)}{\epsilon^2g(f)},\]
which completes the proof.
\end{upperboundproof}

\section{Computing $\ps$}\label{sec: computing ps}

The value $\ps$  is a parameter that is needed for Algorithm~\ref{algo: approximate}.
That means we need a way to compute it for any decreasing function.
As we mentioned before, the only penalty for overestimating $\ps$ is inflation of the storage used by the algorithm so to over-estimate $\ps$ by a constant factor is acceptable.
This section shows that one can find $\ps'$ such that $\ps\leq\ps'\leq 4\ps$ quickly and by evaluating $g$ at just $O(\log{m})$ points. 

Because $g$ is decreasing, the maximum of~\eqref{eq: ps definition} will be achieved by a vector $f$ of length $m$.
Lemma~\ref{lem: decreasing lower bound booster} says that we might as well take all of the other frequencies to be equal, so we can find a near maximizer by enumerating the single value of those frequencies.
Specifically, 
\[s(y) = \min\left\{\frac{m}{y},\frac{g(1)}{\epsilon g(y)}\right\}\]
is the maximum bound we can achieve using $y$ as the single frequency.
The value of $\ps$ is at most twice $\max\{s(y):(m/n)\leq y\leq m\}$, by Lemma~\ref{lem: decreasing lower bound booster}.

But we do not need to check every $y=1,2,\ldots,m$ to get a pretty good maximizer.
It suffices to check only values where $y$ is a power of two.
Indeed, suppose that $y^*$ maximizes $s(y)$ and let $y^*\leq y'\leq 2y^*$.
We will show that $s(y')\geq s(y^*)/2$, and since there is a power of two between $y^*$ and $2y^*$ this implies that its $s$ value is at least $s(y^*)/2\geq \ps/4$.

Since $y^*$ is a maximizer we have $s(y')\leq s(y^*)$, and because $y'\geq y^*$ and $g$ is decreasing we have $g(y')\leq g(y^*)$.
This gives us
\[\frac{g(1)}{\epsilon g(y')}\geq \frac{g(1)}{\epsilon g(y^*)}\geq s(y^*).\]
We also have
\[\frac{m}{y'}\geq \frac{m}{2y^*}\geq \frac{1}{2}s(y^*).\]
Combining these two we have $s(y')\geq s(y^*)/2$.

Thus, one can get by with enumerating at most $\lg m$ values to approximate the value of the parameter $\ps$.
Take the largest of the $\lg m$ values tried and quadruple it to get the approximation to $\ps$.



\bibliography{Decreasing}

\appendix

\section{Proof of Lemma~\ref{lem: decreasing lower bound}}\label{app: lower bound proof}

\newenvironment{lowerboundlemmaproof}{\paragraph{Proof of Lemma~\ref{lem: decreasing lower bound}.}}{}

\begin{lowerboundlemmaproof}
Let $s = \lfloor|\supp(f)|/2\rfloor$ and let $\calA$ be an $(1\pm\epsilon)$-approximation algorithm.
The reduction is from $\disj(s,2)$ where Alice receives $A\subseteq [s]$ and Bob receives $B\subseteq[s]$.
Their goal is to jointly determine whether $A\cap B=\emptyset$ or not.
Our protocol will answer the equivalent question: is $B\subseteq A^c$ or not?  
Alice and Bob will answer the question by jointly creating a notional stream, running $\calA$ on it, and thresholding the outcome.

For each $d\in A^c$, Alice puts $(d,1)$ in the stream $y$ times.  
She then runs $\calA$ on her portion of the stream and sends the contents its memory 
to Bob.  
For each $d\in B$, Bob adds $(d,1)$ to the stream.
Bob runs $\calA$ on his portion of the stream and sends the memory back to Alice.
She recreates her portion of the stream, advances $\calA$, sends the memory to Bob, etc., until each player has acted $k$ times.
In addition to the algorithm's memory, on each pass Alice sends at most $\lceil k^{-1}\lg|A|\rceil$ binary digits of $|A|$ so that Bob knows $|A|$ at the end of the protocol.

The stream is a member of $\calI$ by construction; let $f'$ be its frequency vector.
At the end, Bob finishes computing $\calA(f')$.
All of the frequencies are $y$, $y+1$, or $1$.
If 
\[\calA(f')\leq (1+\epsilon)[|B|g(y+1) + (s-|A|-|B|)g(y)],\] 
then Bob declares $B\subseteq A^c$ and otherwise $B\not\subseteq A^c$.

The exact value of $g(f')$ is 
\[|A\cap B|g(1) + |B\setminus A|g(y+1) + (s-|A|-|B|+|A\cap B|)g(y).\]
If $B\subseteq A^c$ this value is
$$V_0 := |B|g(y+1) + (s-|A|-|B|)g(y),$$
and otherwise, because $g$ is decreasing, it is at least
$$V_1 := g(1) + (|B|-1)g(y+1) + (s-|B|-|A|+1)g(y).$$ 
We find
\[V_1 - V_0
  \geq g(1)
  \geq \epsilon g(f)
  \geq 2\epsilon sg(y)
  \geq 2\epsilon V_0\]
Hence, if $\calA(f')$ is a $(1\pm\epsilon)$-approximation to $g(f')$, then Bob's decision is correct.
The protocol with solves $\disj(s)$ which requires, in the worst case, $\Omega(s)$ bits of communication including $O(k^{-1}\lg s)$ bits to send $|A|$ and $\Omega(s)=\Omega(|\supp(f)|)$ bits for $(2k-1)$ transmissions of the memory of $\calA$.
Thus, in the worst case, at least one transmission has size $\Omega(|\supp(f)|/k)$.
\end{lowerboundlemmaproof}

\section{Details of the algorithm}\label{app: details}

First, we prove Lemma~\ref{lem: ps relations}, which is used in the proof of correctness of the algorithm.

\newenvironment{psrelationsproof}{\paragraph{Proof of Lemma~\ref{lem: ps relations}.}}{}
\begin{psrelationsproof}
Let $\ps_\epsilon=\ps(\epsilon,g,m,n)$ and define $\ps_\alpha$ similarly.
Let $f\in\calF$ such that $\ps_\alpha = |\supp(f)|$ and $g(f)\leq\alpha^{-1}g(1)$, without loss of generality the coordinates are ordered such that $f_1\geq f_2\geq\cdots\geq f_{\ps_\alpha}>0$.
Let $s' = \frac{\alpha}{\epsilon} \ps_\alpha$, and let $f'$ be the vector that takes the first $\lfloor s'\rfloor$ coordinates from $f$ and is $0$ thereafter.
The choice is made so that $f'\in\calF$ and 
\begin{align*}
g(f')\leq\frac{\alpha}{\epsilon} g(f)\leq \epsilon^{-1}g(1).
\end{align*}
Then, by definition of $\ps_\epsilon$, we have
\begin{align*}
\ps_\epsilon \geq |\supp(f')| = \left\lfloor\frac{\alpha}{\epsilon} \ps_\alpha\right\rfloor\geq \frac{\alpha}{\epsilon} \ps_\alpha -1.
\end{align*}
\end{psrelationsproof}

The streaming implementation in the turnstile model will make use of the \textsc{Count Sketch} algorithm of Charikar, Chen, and Farach-Colton~\cite{charikar2002finding}.
It is easy to adapt their algorithm for the purpose of finding $\supp(f)$.
This gives us the following theorem.

\begin{theorem}[Charikar, Chen, Farach-Colton~\cite{charikar2002finding}]\label{thm: count sketch}
  Suppose that $S$ is a stream with at most $s$ items of nonzero frequency.
  There is a turnstile streaming algorithm~\textsc{Count Sketch}$(S,s,\delta)$ using $O(s\log{\frac{n}{\delta}}\log{M})$ bits that, with probability at least $1-\delta$, returns all of the nonzero frequencies in $S$.
\end{theorem}

The sampling algorithm follows.  
Since we do not know $|\supp(f)|$ at the start of the stream, we guess $O(\log{n})$ possible values for it and try each one.
After parsing the entire stream, we can use an estimate of $L^0=|\supp(f)|$ in order to determine which guess is correct.
We use $\widehat{L^0}(S^{(i)},\epsilon,\delta)$ to denote the output of an algorithm that produces a $(1\pm\frac{1}{8})$-approximation to $L^0$ with probability at least $1-\delta$, for example the algorithm of Kane, Nelson, and Woodruff~\cite{kane2010optimal}.
After the formal presentation of the algorithm we prove its correctness and the claimed storage bounds.

\begin{algorithm}
  \begin{algorithmic}[1]
    \Procedure{\textsc{Sketch}}{Stream $S$, $s>0$}
    \State $\ell\gets \lceil\lg(n/s)\rceil$
    \For{$0\leq i\leq \ell$}
    \State Sample pairwise independent r.v.s $X_{i,d}\sim\Bernoulli(2^{-i})$, for $d\in[n]$
    \State Let $S^{(i)}$ be the substream of $S$ with items $\{d: X_{i,d}=1\}$
    \State $U^{(i)}\gets \textsc{Count Sketch} (S^{(i)},96s,1/48)$
    \EndFor
    \State $L\gets \widehat{L^0}(S^{(i)},1/8,1/18)$
    \State $i^*\gets \max\left\{0,\left\lfloor\lg\frac{L}{18s}\right\rfloor\right\}$
    \State \Return{$U^{(i^*)},q=2^{-i^*}$}
    \EndProcedure
  \end{algorithmic}
  \caption{Pairwise independent sampling with probability $q\geq s/|\supp(f)|$.}
  \label{algo: SimpleSketch}
\end{algorithm}

\begin{theorem}\label{thm: SimpleSketch correctness}
With probability at least $7/9$, Algorithm~\ref{algo: SimpleSketch} samples each item in $\supp(f)$ with probability $q\geq s/|\supp(f)|$ and the resulting sample of size $O(s)$.
The algorithm can be implemented with $O(s\log(M)\log^2(n))$ bits of space.
\end{theorem}
\begin{proof}
Let 
\[k= \left\lfloor\lg \frac{|\supp(f)|}{16s}\right\rfloor.\]
If $i^*\in\{k-1,k\}$, the streams $S^{(k-1)}$ and $S^{(k)}$ both have small enough support, and the two outputs $U^{(k-1)}$ and $U^{(k)}$ of \textsc{Count Sketch} are correct, then the output is correct.
We show that the intersection of these events occurs with probability at least $7/9$.

First, with probability at least 17/18 $L$ is $(1\pm 1/8)$-approximation to $|\supp(S)|$.
A direct calculations then shows that $i^*\in\{k-1,k\}$.

The following two inequalities arise from the definition of $k$ 
\begin{equation}\label{eq: k inequality}\frac{64s}{|\supp(f)|}\geq 2^{-(k-1)}\geq 2^{-k}\geq \frac{16s}{|\supp(f)|}.\end{equation}
The first inequality implies that the expected support sizes of $S^{(k-1)}$ and $S^{(k)}$ and their variances are all at most $64s$.  
Chebyshev's inequality implies that each of these values exceeds $96s$ with probability no larger than $64/32^2 = 1/16$.
So long as they don't, both streams are valid inputs to $\textsc{Count Sketch}$.
The last inequality of \eqref{eq: k inequality}, with Theorem~\ref{thm: count sketch}, implies that the sampling probability is correct.

Putting it together, the total probability of failure is no larger than 
\begin{equation}\label{eq: decreasing turnstile final union bound}
\frac{1}{18} + \frac{2}{16} + \frac{2}{48}=\frac{2}{9},
\end{equation}
where the terms come from the $|\supp(f)|$ estimation, the support sizes of substreams $k-1$ and $k$, and \textsc{Count Sketch}.

The space bound for turnstile streams follows from Theorem~\ref{thm: count sketch}.
Approximating the support size of the stream with $\widehat{L^0}$ can accomplished with $O(\log n\log\log nM)$ bits using the algorithm of Kane, Nelson, and Woodruff~\cite{kane2010optimal}.
\end{proof}

Because of deletions in the turnstile model, we need to wait until the end of the stream to rule out any of the guesses of $|\supp(f)|$.
This is not the case in the insertion only model.
As soon as the number of nonzero counters grows too large we can infer that the sampling probability is too large and discard the sample.
It turns out that doing so is enough to cut a $\log{n}$ factor from the space complexity of Algorithm~\ref{algo: SimpleSketch}.
A further $\log{n}$ factor can be saved because \textsc{Count Sketch} is not needed in the insertion-only model.

\begin{corollary}\label{cor: insertion-only SimpleSketch}
Algorithm~\ref{algo: SimpleSketch} can be implemented with  $O(s\log M + \log^2n)$ bits of storage for insertion-only streams. 
\end{corollary}
\begin{proof}
Define $\ell$ independent collections of pairwise independent random variables $Y_{i,d}\sim\Bernoulli(1/2)$, for $d\in[n]$, and choose the random variables in the algorithm to be 
\[X_{i,d}= \prod_{j=1}^i Y_{i,d}.\]
One easily checks that each collection $\{X_{i,d}\}_{d\in[n]}$ is pairwise independent and that $P(X_{i,d}=1)=2^{-i}$, for all $i$ and $d$.
Storing the seeds for the collection $Y_{i,d}$ requires $O(\log^2 n)$ bits.

We can first save a $\log n$ factor by bypassing $\textsc{Count Sketch}$ and instead simply storing counters for each element that appears in each of the $\ell$ substreams.
The counters should be stored in a hash table or other data structure with no space overhead and a small look-up time.
Let us label the maximum number of counters to be stored for each substream as $t$.
We choose $t = \max\{96s, \ell\}$.
If the set of counters for each substream is discarded as soon as the number of nonzero counters exceeds the limit of $O(t)$, then the total storage cannot grow to large.

According to Lemma~\ref{lem: only a few survive}, the algorithm uses more than $12t$ counters with probability at most $1/6\ell$, at any given instant.

For each $0\leq i\leq\ell$ let $T^{(i)}$ be the longest prefix of stream $S^{(i)}$ such that $|\supp(T^{(i)})|\leq s$ and let $k^{(i)}$ denote the number of updates in $T^{(i)}$. 
Now, notice that the number of counters stored locally maximum at each $k^{(i)}$ and increasing for updates between $k^{(i)}$ and $k^{(i+1)}$.
Thus, it is sufficient to bound the storage used by the algorithm at these points.

By a union bound, the probability that the number of counters used by the algorithm at any point $k^{(1)},k^{(2)},\ldots,k^{(\ell)}$ is more than $12t$ is at most $\ell\cdot 1/6\ell = 1/6$.
Finally, adapting the final union bound of \eqref{eq: decreasing turnstile final union bound} in the previous proof we have that the probability of error is at most $(1/18)+(1/6)=2/9$.
\end{proof}

\begin{lemma}\label{lem: only a few survive}
Let $v\in\{0,1\}^{n}$, define $\ell$ independent collections of pairwise independent random variables $Y_{i,d}\sim\Bernoulli(1/2)$, for $s\in[n]$ and $i\in[\ell]$, and set 
\[X_{i,d}= \prod_{j=1}^i Y_{i,d}.\]
For a given $s\in\N$, set $k=0$ if $\sum_dv_d\leq s$ or $k=\max\{i: v^TX_i> s\}$ otherwise, where $X_i=(X_{i,1},X_{i,2},\ldots,X_{i,n})\in\{0,1\}^n$. 
Then
\[P(\sum_{i=k+1}^\ell v^TX_i > 4s)\leq \frac{1}{2s}.\]
\end{lemma}
\begin{proof}
The sum is clearly monotonically increasing, so without loss of generality assume $\ell=\infty$.
Notice that if $k>0$, the sum is unchanged (i.e., it remains the same random variable) upon replacing $v$ with the coordinate-wise product of $v$ and $X_{k}$. 
Thus we may also assume that $k=0$, i.e., $|\supp(v)|\leq s$.

For each $d\in\supp(v)$, let $Z_d = \sup\{i:X_{i,d}=1\}$.  
Notice that $\{Z_d\}_{d\in\supp(v)}$ is a pairwise independent collection of $\Geometric(1/2)$ random variables and let $Z = \sum_{d\in\supp(v)}Z_d$.  We have that
\[Z = \sum_{i=0}^\infty v^TX_i,\]
because $X_{i,d}=0$ implies $X_{j,d}=0$ for all $j>i$.

Pairwise independence implies $EZ =Var(Z)=2|\supp(v)|\leq 2s$, and by Chebyshev's inequality
\[P(|Z-2s|> 2s)\leq \frac{Var(Z)}{4s^2}\leq \frac{1}{2s}.\]
\end{proof}

\end{document}